\documentclass[aps,pra,prereprint,superscriptaddress,notitlepage]{revtex4-1}

\usepackage{epsfig,graphics,graphicx}
\usepackage{dcolumn}
\usepackage{bm}
\usepackage{amsmath,amssymb,amsthm}
\usepackage{centernot}
\usepackage{bbold}
\usepackage{soul}
\usepackage{mathtools}
\usepackage{fullpage}
\usepackage{units}
\usepackage[usenames,dvipsnames]{xcolor}
\usepackage{braket}
\usepackage{indentfirst}
\usepackage{lipsum}
\usepackage[export]{adjustbox}

\usepackage[utf8]{inputenc}
\usepackage[T1]{fontenc}

\usepackage{dsfont}

\usepackage{mathptmx}
\usepackage{amssymb}
\usepackage{amsmath}
\usepackage{latexsym}
\usepackage{amsfonts}
\usepackage{soul}

\usepackage{hyperref}
\hypersetup{pdfpagemode=UseNone}

\newtheorem{theorem}{Theorem}
\newtheorem{definition}[theorem]{Definition}

\newtheorem{lemma}[theorem]{Lemma}
\newtheorem{proposition}[theorem]{Proposition}
\newtheorem{corollary}[theorem]{Corollary}

\newcounter{rem}
\newcommand{\mc}[1]{\mathcal{#1}}

\def\>{\rangle}
\def\<{\langle}

\newcommand{\borb}[2]{\left | #1 \rangle\!\langle #2 \right |}
\renewcommand{\rho}{\varrho}

\def\textbf#1{{\bf #1}}

\def\beq{\begin{equation}}
\def\eeq{\end{equation}}
\def\beqa{\begin{eqnarray}}
\def\eeqa{\end{eqnarray}}
\def\eea{\end{array}}
\def\bea{\begin{array}}
\newcommand{\bei}{\begin{itemize}}
\newcommand{\eei}{\end{itemize}}
\newcommand{\bee}{\begin{enumerate}}
\newcommand{\eee}{\end{enumerate}}
\def\bep{\begin{proposition}}
\def\eep{\end{proposition}}
\def\bel{\begin{lemma}}
\def\eel{\end{lemma}}
\def\bet{\begin{theorem}}
\def\eet{\end{theorem}}
\def\bed{\begin{definition}}
\def\eed{\end{definition}}

\definecolor{cgreen}{RGB}{26, 199, 76}
\newcommand{\cris}[1]{{\color{black} #1}}

\usepackage{soul}
\definecolor{violeta}{cmyk}{0.07,0.90,0,0.34}

\newcommand{\bo}[1]{{\color{black} #1}}
	
\begin{document}
\title{Small violations of Bell inequalities for multipartite pure 
random states}

\author{Cristhiano Duarte}
\email{cduarte@iip.ufrn.br}
\thanks{These two authors contributed equally.}
\affiliation{International Institute of Physics, Federal University of Rio 
Grande do Norte, 59078-970, P. O. Box 1613, Natal, Brazil}
\affiliation{Departamento de Matem\'{a}tica, Instituto
de Ci\^{e}ncias Exatas, Universidade Federal de Minas Gerais, CP
702, CEP 30123-970, Belo Horizonte, MG, Brazil.}

\author{Raphael C. Drumond}
\email{raphael@mat.ufmg.br}
\thanks{These two authors contributed equally.}
\affiliation{Departamento de Matem\'{a}tica, Instituto
de Ci\^{e}ncias Exatas, Universidade Federal de Minas Gerais, CP
702, CEP 30123-970, Belo Horizonte, MG, Brazil.}

\author{Roberto I. Oliveira}
\email{rimfo@impa.br}
\affiliation{Instituto Nacional de Matem\'atica Pura e Aplicada-IMPA 
Estrada 
Dona 
Castorina, 110, Jardim Bot\^anico 22460-320, Rio de Janeiro, RJ, Brazil}

\date{\today}


\begin{abstract}
For any finite number 
of parts, measurements and outcomes in a Bell scenario we estimate 
the 
probability of random 
$N$-qu$d$it pure states to substantially violate any Bell 
inequality with \cris{uniformly} bounded coefficients. We prove that under some 
conditions on the 
local dimension the probability to find any significant amount of violation goes to 
zero exponentially fast as the number of parts goes to infinity. \cris{In addition, 
we also prove that} if the 
number of parts is at least $3$, this probability also goes to zero \cris{as the  
the local Hilbert space dimension goes to infinity}.
\end{abstract}


\maketitle
\section{Introduction}

In his seminal paper~\cite{Bell66}, John Bell proved that, under certain assumptions, there exist 
predictions of quantum theory incompatible with those of any physical theory 
described by a Local Hidden Variable (LHV) 
Model~\cite{Loubenets12,Loubenets17,BCPSW13}. Roughly speaking, 
these are models in which outcomes of spacelike separated 
measurements 
are independent conditionally on the knowledge of an underlying hidden 
variable. Bell's key contribution was to show that LHV models must satisfy certain linear 
inequalities on joint probabilities that are now called \emph{Bell 
Inequalities}, which quantum mechanics might violate. This is
one of most important results in quantum physics~\cite{Cabello96} so far, with deep 
implications to our knowledge of the world~\cite{VB14,RBG14,AMO09} and growing interest in 
practical 
applications~\cite{BCPSW13,Vicente14,Barrett05b}.  
It is thus natural to ask how difficult it is to construct a 
quantum experimental setup that violates one of 
these inequalities~\cite{Aspect75,Delft-free-15,Delft-free-16,FC72,Winter14}.

The present paper answers one variant of 
that question. Using probabilistic 
techniques~\cite{GHJPV2015,GLPV17,PalazuelosYin15,Palazuelos15,Loubenets12,Loubenets17,
AS04,BrietVidick13,LNBR10}, and generalizing previous results of the 
authors in~\cite{DO12}, we are able to find an upper-bound for the 
typical behaviour of optimal violations for any Bell scenario $\Gamma=(N,m,v)$. We also discuss
the effect of the size of the local dimension $d$ \cite{GHJPV2015} on the 
probability that a $N$-partite $d$-dimensional quantum system violates any 
Bell inequality associated with this scenario. We find that\cris{, under some 
conditions,} typical pure states \emph{do not} 
produce large violations of Bell inequalities, in spite of the fact that they are likely to be
highly entangled~\cite{BCPSW13, Loubenets17}. This apparent paradox sheds light on 
the important difference (not always made clear) between the 
concepts of entanglement~\cite{HHHH09} and Bell non-locality~\cite{BCPSW13}.

The paper is organized as follows. Section~\ref{sec:BellInequalities} 
is devoted to set up some notation and definitions that will be used 
in the remainder of the paper. In Section~\ref{sec:SmallProbabilitiesHighViolations} we 
discuss our main result, in the form of Theorem~\ref{Thm:TheTheorem}, and its consequences. In 
Section~\ref{sec:TheProof} and 
Section~\ref{sec:PuttingAllTog} we give all key 
ingredients necessary to prove our main theorem. Section~\ref{sec:Conclusion} 
presents open questions, further works as well as our conclusions.


\section{Bell Inequalities}
\label{sec:BellInequalities}

\subsection{Basic Definitions}
\label{subsec:BellInequalitiesDefinitions}

As our starting point in the paper, we consider multipartite 
device-independent, or 
black-box, scenarios~\cite{BCPSW13,PSV16}. That is, we consider a 
general correlation 
scenario, denoted by 
	\begin{equation}
        \Gamma:=(N,m,v).
        \label{eq:DefScenario}
	\end{equation}
The triple $\Gamma$ describes a scenario with the following characteristics:
\begin{itemize}
\item  $N$ black-boxes are distributed among $N$ players;
\item each of these boxes 
admits $m$ different inputs; and
\item for each input, among all $v$ possible 
distinct outputs, only one outcome is observed given that choice of 
input~\footnote{Notice that there is no loss of generality in supposing that every 
box has access to the same number $m$ of inputs as well as to the same number $v$ of 
outputs.This abstraction could be realized 
considering extra buttons that are never 
pushed and additional light bulbs that are never turned on.}. 
\end{itemize}

Without having direct access to \cris{all} internal details of each box, the best 
description one 
has for $\Gamma$ is through the correlations among inputs and outputs. This means that 
such 
experimental scenarios \cris{should} be described 
by a vector, $\vec{p} \in \mathds{R}^{(vm)^N}$ usually called 
\emph{behaviour}~\cite{BCPSW13,BorisTsirelson93}. The components of \cris{such a} 
vector are given by numbers
\begin{equation}
p_{a_1,...,a_N,x_1,...,x_N}=P(a_1,...,a_N|x_1,...,x_N),
\end{equation}
corresponding to the joint probability of obtaining the outcome list 
$(a_1,...,a_N) \in [v] \times ... \times [v]$ given the inputs $(x_1,...,x_N) 
\in [m] \times ... \times [m]$. 

An \emph{admissible behaviour} is a vector $\vec{p} \in 
\mathds{R}^{(mv)^{N}}$ that satisfies the following constraints:
\begin{align}
 &P(a_1,...,a_N|x_1,...,x_N) \in [0,1], \forall \,\, a_1,...,a_N, \forall \,\, 
x_1,...,x_N \label{eq:defConstraintsInBehavioursPosit} \\
 &\sum_{a_1,...,a_N}P(a_1,...,a_N|x_1,...,x_N)=1, \,\, \forall \,\, x_1,...,x_m.
 \label{eq:defConstraintsInBehavioursSum}
\end{align}
We let $\mathcal{B}_{\Gamma}$ denote the set of \emph{all admissible behaviours} associated with 
the correlation scenario $\Gamma=(N,m,v)$.

\subsection{Non-Signalling and Local Hidden Variable Constraints}

A behaviour $\vec{p} \in \mc{B}_{\Gamma}$ satisfies the \emph{non-signaling 
constraint} when for 
all choices of subsets $\mc{I} \subset [N]$, say $\mc{I}=\{i_1,...,i_k\}$, one has
\begin{equation}
  \sum_{\substack{a_i \\ i \in \mc{I}}}P(a_1,...,a_N|x_1,...,x_N)= 
P(a_{i_1},...,a_{i_k}|x_{i_1},...,x_{i_k})=\sum_{\substack{a_i \\ i \in 
\mc{I}}}P(a_1,...,a_N|x_1^{\prime},...,x_N^{\prime}), \forall \,\, a_{i_1},...,a_{i_k}
\label{eq:DefNonSignalling}
\end{equation}
for all possible inputs $(x_1,...,x_N)$ and $(x_1^{\prime},...,x_N^{\prime})$ with 
$x_{i_1}=x_{i_1}^\prime$, $\dots$, $x_{i_k}^\prime = x_{i_k}$. The set of all such 
behaviours is denoted by $\mc{NS}$.

 We say that a behaviour $\vec{p} \in 
\mathds{R}^{(vm)^{N}}$ admits a \emph{local hidden variable} (LHV) model when there 
exists a probability space $\mc{P}=(\Omega,\Sigma,\mu)$, and response functions 
\begin{align}
p(a_i|x_i, \cdot): \Omega &\longrightarrow [0,1] \nonumber \\
                   \omega &\mapsto p(a_i|x_i, \omega)
\end{align}
such that, for every input and output one has~\footnote{For 
more details, including historical remarks, we refer to 
references~\cite{BCPSW13,Loubenets12,Loubenets17,PalazuelosYin15,Palazuelos15}} :
\begin{equation}
 P(a_1,...,a_N|x_1,...,x_N)= \int_{\Omega}p(a_1|x_1, \omega)p(a_2|x_2, 
\omega)...p(a_N|x_N, \omega)\mu(d\omega).
 \label{eq:DefLocalConstraints}
\end{equation}
The set of all behaviours $\vec{p} \in \mc{B}_{\Gamma}$ which admit a LHV model is 
called \emph{local set} and denoted by $\mc{L}$. 
\cris{Notice that any admissible behaviour which admits a LHV model 
also satisfies the non-signalling constraints. In other words, the set of LHV models 
is contained inside the set of non-signalling behaviours.}

In this paper we only consider scenarios in which the number of parts, the number of inputs and the 
number 
of outcomes are all finite. Therefore, we can assume~\cite{BCPSW13} that the underlying 
probability space is also finite, and both $\mc{L}$ and $\mc{NS}$ 
are polytopes~\cite{LexSch03,Pitowsky91}, the former with number 
of vertices $a=v^{mN}$ and affine dimension $D=(mv)^N$.

\subsection{Bell Inequalities (with bounded coefficients)}
\label{sec:bellineq}

A \emph{Bell inequality} is defined through a linear functional $T$ \cris{acting on} 
a behaviour \cris{$\vec{p}$}:
\begin{equation}
T(\vec{p})=\sum_{\substack{ a_1,...,a_N \\ x_1,...,x_N}} T_{a_1,...,a_N|x_1,...,x_N} 
P(a_1,...,a_N|x_1,...,x_N),
 \label{eq:DefBellFunc}
\end{equation}
where $T_{a_1,...,a_N|x_1,...,x_N}$ are real numbers. \cris{Defining} 
$\Delta_{i}=\text{inf}_{\vec{p}\in\mc{L}}T(\vec{p})$ and 
$\Delta_{u}=\text{sup}_{\vec{p}\in\mc{L}}T(\vec{p})$, then
\begin{align}
&\sum_{\substack{ a_1,...,a_N \\ x_1,...,x_N}} T_{a_1,...,a_N|x_1,...,x_N} 
P(a_1,...,a_N|x_1,...,x_N)\leq \Delta_{u},\text{ and } \label{bellineqa}\\
 &\sum_{\substack{ a_1,...,a_N \\ x_1,...,x_N}} T_{a_1,...,a_N|x_1,...,x_N} 
P(a_1,...,a_N|x_1,...,x_N)\geq \Delta_{i},\label{bellineqb}
\end{align}
\cris{are tight \emph{Bell inequalities}, in the sense that they are satisfied by all
local behaviours and, in addition, there exists at least one local behaviour 
saturating each bound}. For the sake of simplicity in defining the 
\emph{degree of violation} of such inequalities (see Eq.~\ref{eq:DefDegreeOfViolation} below), but 
without loss of 
generality, we assume 
throughout the paper that 
\begin{equation}
\text{max}\{|\Delta_{i}|,|\Delta_{u}|\}=1.
\end{equation}

We denote the set of all linear functionals defining Bell inequalities (for a fixed scenario) by 
$\mc{T}$. For any $b>0$, we denote by $\mc{T}_{b}$ the subset of $\mc{T}$ where all coefficients 
(components of the linear functional) are bounded by $b$, that is
\begin{equation}
 \mc{T}_{b}=\{T\in\mc{T}:|T_{a_1,...,a_N| x_1,...,x_N}|\leq b \text{ for all 
}a_{1},...,a_{N},x_{1},...,x_{N}\}.\label{boundedcoefs}
\end{equation}

One the one hand, whatever the value that $b$ assumes, $\mc{T}_{b}$ will always be a 
proper 
subset of $\mc{T}$. Indeed, due to the constraints satisfied by the components of behaviours, one 
can insert terms with arbitrarily large coefficients to the functional~\eqref{eq:DefBellFunc} 
without ever altering its value on behaviours. 

On the other hand, as we will argue below, using again these same constraints, one can show that 
every Bell inequality~\eqref{bellineqa}~and~\eqref{bellineqb} defined by a functional is 
equivalent~\footnote{We say that two Bell inequalities are equivalent if a behavior violate one of 
them if and only if it also violate the other. } to one whose functional only has non-negative 
coefficients and belongs to 
$\mc{T}_{1}$. This possibility is often tacitly 
assumed~\cite{CSW10,CSW14,RDLTC14,CGLMP02,CabelloPRA16,Sliwa03} and many notable inequalities 
are of this kind. For instance, all three~\cite{SBBC13} pentagonal 
Bell 
inequalities
\begin{align}
 I^{P}_{1}&=P(00|00) + P(11|01) + P(10|11) + P(00|10) + P(11|00) \leq 2; \\
 I^{P}_{2}&=P(00|00) + P(11|01) + P(10|11) + P(00|10) + P(\_ 1|\_ 0) \leq 2; \\
 I^{P}_{3}&=P(00|00) + P(11|01) + P(10|11) + P(00|10) + P(11|20) \leq 2
 \label{eq:PentBellIneqs}
\end{align}
as well as a symmetric version of the simplest tight Bell inequality (after CHSH)
violated by quantum theory~\cite{Foissart81,CG04,BG2008}:
\begin{align}
I_{3322} =& P(00|01) + P(00|02) + P(00|10) + P(00|12) + P(00|20) + P(00|21) \\
          &+ P(01|11) + P(10|11) + P(11|11) + P(01|22) + P(10|22) + P(11|22) \\
          &+ P(1\_|0\_) + P(1\_|1\_) + P(\_1|\_0) + P(\_1|\_1) \leq 6,
\end{align}
are examples of it (upon normalizing the classical bound).

\bo{In addition, recently the authors in~\cite{CSW10,CSW14,RDLTC14} have made 
a heavy use of the fact that any Bell inequality can be rewritten as an equivalent 
inequality with positive coefficients lying in $\mathcal{T}_{b}$, with $b>0$. Starting from a Bell 
inequality with positive coefficients they associate a graph $\mathcal{G}$ representing that 
inequality. Each event is associated 
with a vertex, and two vertices are connected whenever their underlying events are 
exclusive~\cite{RDLTC14}. The weight of each vertex in that representation being given 
by the positive coefficient associated with the underlying event in the Bell 
inequality. Having constructed such a weighted-graph, the authors have shown that $i)$ 
the classical bounded for the original inequality is given by $\alpha(\mathcal{G})$ 
the independence number of $\mathcal{G}$, that $ii)$ the maximal quantum violation is 
$\vartheta(\mathcal{G})$ the Lovász number of $\mc{G}$, and that $iii)$ the maximum 
value admissible for more general probabilistic theories is $\alpha^{\ast}(\mc{G})$ 
the fractional-packing number of $\mc{G}$. It shows that such rewritings are not only aesthetic 
tricks, they do have provided fruitful results for the analysis of violations of Bell 
inequalities.}

Indeed, suppose we are given a Bell inequality~\eqref{bellineqa} and assume $\Delta_{u}>0$. 
If
\begin{equation}
\mathcal{I}=\{(a_1,...,a_N,x_1,...,x_N) 
\in [v]^{N}\times[m]^{N}: \,\, T_{a_1,...,a_N \vert 
x_1,...,x_N} < 0 \},
\end{equation}
for each element $(a_{1}^{\prime},...,a_{N}^{\prime},x_{1}^{\prime},...,x_{N}^{\prime}) \in 
\mathcal{I}$ we can use the fact that probabilities are normalized to write
\begin{equation}
P(a_{1}^{\prime},...,a_{N}^{\prime}| 
x_1,...,x_N)=1 \,\,\, - \sum_{a_1 \neq a_{1}^{\prime},...,a_N \neq 
a_{N}^{\prime}}P(a_1,...,a_N|x_1,...,x_N),\end{equation} and change all terms with
negative coefficients appearing in Eq.~\eqref{bellineqa} by a sum of terms 
with positive coefficients plus some negative constants. Leaving all terms 
with probabilities on the \emph{l.h.s} of the inequality and all constants on the 
\emph{r.h.s}, the new inequality has
classical bound $\Theta=\Delta_{u} - \sum_{(a_1,...,a_N \vert 
x_1,...,x_N) \in \mathcal{I}}T_{a_1,...,a_N \vert 
x_1,...,x_N}>0$. Dividing both sides by $\Theta$, we obtain an 
inequality
\begin{equation}
 \sum_{\substack{ a_1,...,a_N \\ x_1,...,x_N}} \tilde{T}_{a_1,...,a_N|x_1,...,x_N} 
P(a_1,...,a_N|x_1,...,x_N) \leq 1,
 \label{eq:DefBellIneqPositiveBounded}
\end{equation}
whose coefficients are non-negative and, therefore, bounded by 
$1$ (otherwise it would be possible to construct a LHV theory
violating the inequality). A similar trick can be done if $\Delta_{u}<0$ and to 
inequality~\eqref{bellineqb}.

\subsection{Degree of violations by quantum states}

A test of a Bell inequality~\eqref{bellineqa}~or~\eqref{bellineqb} can be designed for 
a quantum 
system by choosing $m$ POVM's with $v$ outcomes, $\{\Pi_{a_k,x_k}^{k}\}_{a_k=1}^{v}$, for each 
part $k$. We denote by $A$ the collection of these POVM's and define the so-called Bell 
operator
\begin{equation}
 \mathfrak{B}_{T,A}:= \sum_{\substack{ 
a_1,...,a_N \\ x_1,...,x_N}} T_{a_1,...,a_N|x_1,...,x_N}\Pi_{a_1,x_1}^{1} \otimes 
... \otimes \Pi_{a_N,x_N}^{N}.
 \label{eq:BellOperators}
\end{equation}
Moreover, we denote by $\mc{A}$ the set of all possible POVM's choices for a fixed scenario 
and $N$-fold tensor 
product of Hilbert spaces. A (possible) violation for the Bell inequality is then evaluated through 
the function:
\begin{equation}
 Q(\ket{\psi},T,A)=\text{Tr}(\mathfrak{B}_{T,A} \borb{\psi}{\psi}).
 \label{eq:AlmostBornRule}
\end{equation}
Whenever $|Q(\ket{\psi},T,A)|>1$ one of Bell inequalities with coefficients $T$ is 
violated by 
the local measurements described by $A$. \cris{For that reason, it makes sense that 
the}  \emph{degree of 
violation}~\cite{Loubenets12,Loubenets17,PalazuelosYin15,Palazuelos15} should be defined 
by the quantity 
\begin{equation}
|Q(\ket{\psi},T,A)|. 
\label{eq:DefDegreeOfViolation} 
\end{equation}

It will be crucial for our results the bounds on quantum violations of Bell inequalities 
found by Loubenets in~\cite{Loubenets17}, which guarantees that
\begin{equation}
 |Q(\ket{\psi},T,A)|\leq (2m-1)^{N} \label{cotaloubenets}
\end{equation}
\cris{for all $T$ and $A$ associated with a fixed scenario, and all $\ket{\psi}$ 
of any $N$-fold tensor product Hilbert space.}


\section{Small Probabilities of High Violations}
\label{sec:SmallProbabilitiesHighViolations}

It is already known~\cite{BCPSW13,Loubenets17,Palazuelos15} that 
if either the local dimension $d$ of each quantum system, or the number $N$ of parts, 
is sufficiently large, then typical quantum states are highly entangled. Should we expect, then, 
that typically there will be at least one Bell inequality that is greatly violated? 
\cris{Our main objective in the present work is, then, to approach the following question:}

\begin{center}
 \textit{given a typical pure state $\ket{\psi}$ composed by $N$ $d-$dimensional 
quantum systems, what should one expect for its largest possible violation over all relevant Bell 
inequalities in a given scenario $\Gamma=(N,m,v)$?}
\end{center}
In our framework, optimal violations (if any) of Bell inequalities  
exhibited by a quantum state $\ket{\psi} \in (\mathds{C}^{d})^{\otimes N}$ are given by the 
functional (see Eq.~\eqref{eq:AlmostBornRule}):
\begin{equation}
 V_{\text{opt}}(\ket{\psi}):= \sup_{\substack{A \in 
\mc{A} \\ T \in \mc{T}_{b} }}  |Q(\ket{\psi},T,A)|,
\label{eq:DefFunctionalOptimizing}
\end{equation}
where the supremum is taken over  all 
quantum implementations of all Bell inequalities whose coefficients are \emph{uniformly 
bounded by $b$}. 

{\textbf{Remark: }Our restriction to optimize the violations over 
$\mc{T}_{b}$ has 
basically a technical motivation: we need to assume it in order to successfully apply our 
techniques. On the one hand, as we have pointed out in Sec~\ref{sec:bellineq}, it holds that for any 
fixed $b>1$, every possible Bell inequality is 
equivalent to one obtained by a functional in $\mc{T}_{b}$. At first sight this fact apparently 
makes irrelevant our restriction. On the other hand, for a fixed quantum state, its optimal 
\emph{degree} of violation may be different even for two \emph{equivalent} Bell inequalities (as 
one can test on simple examples). We do not know, for instance, how to rule out the existence of a 
family of pure quantum states for $N$ parties whose optimal violations only 
takes place on Bell inequalities whose coefficients becomes arbitrary large with $N$. The fact is 
we are not sure of such a restriction to bounded coefficients is relevant or not. However, even if 
it is, we recall from the discussion in section~\ref{sec:bellineq} that we still would be able to 
encompass a huge class of relevant Bell inequalities.}

Now, if $\ket{\psi}$ is a point of a sample space, since $V_{\text{opt}}$ is a 
function of 
$\ket{\psi}$ we 
can consider it as a random variable. Formally, what we would 
like to estimate is
the distribution function of such variable, that is: 
\begin{equation}
\mathds{P}(V_{\text{opt}} > c).
\end{equation}
It is known~\cite{PR92,CASA11} that any entangled 
$N$-partite pure state violate \emph{some} Bell inequality. Since these states have full 
measure on the sphere of pure states, its is clear then that $\mathds{P}(V_{\text{opt}} > c)=1$ for 
$c\leq 1$. For arbitrary $c>1$ we have our main result:
\begin{theorem}
 Given $N,d\geq 2$ integers. Let $\ket{\psi} \in \left( 
\mathds{C}^d \right)^{\otimes N}$ be a unit vector distributed according to the uniform 
measure in the sphere $S_{2d^{N}-1}$ of $\left( \mathds{C}^d \right)^{\otimes N}$, then:
\begin{align}
 \mathds{P}(V_{\text{\emph{opt}}}>c) \leq 4  
 \left[\frac{8bN(mv)^Nd^2}{\delta}+2\right]^{mvNd^{2}+(mv)^N} \times 
e^{-\left(\frac{2d^{N}(c-\delta-1)^2}{36\pi^{2}(2m-1)^{2N-2}}\right)},
\label{eq:MainInequality}
\end{align}
for any $b,\delta>0, \,\, c>\delta+1$. 
\label{Thm:TheTheorem}
\end{theorem}
This theorem allows us to answer the question posed at the beginning of this section  
negatively. A typical state $\ket{\psi}$ composed by $N$ $d$-dimensional quantum systems, with $N$ 
and/or $d$ large enough, \emph{does not} exhibit a significant degree of violation for any 
Bell inequality {(with uniformly bounded coefficients)}. 

In fact, note that it is possible to rewrite Eq.~\ref{eq:MainInequality} as follows 
(see Appendix~\ref{sec:cotaprob}):
\begin{small}
   \begin{align}
   \mathds{P}(V_{\text{opt}}>c) \leq 4 \exp \left\lbrace mvN^{2}d^{2}\log(mv) 
+mvNd^{2}\log\left(\frac{16bNd^{2}}{\delta} 
\right)  \right. + & (mv)^{N}\log\left(\frac{16Nd^{2}}{\delta}\right) + 
N(mv)^{N}\log(bmv)
\nonumber \\
& \left.     
 -\frac{(c-\delta-1)^{2}(2m-1)^2}{18\pi^{2}}\left[\frac{d} { 
(2m-1)^2 
} 
\right]^{N}\right\rbrace. \label{cotadetalhada}
\end{align}
\end{small}
Now, on the one hand, if the local dimension $d$ of each subsystem satisfies
\begin{equation}
d>mv(2m-1)^{2}, 
\end{equation}
and if in addition the uniform bound $b$ is not too large, \emph{e.g.} 
if 
\begin{equation}
 b=\mc{O}\left((mn)^{N}\right),
 \label{eq:BoundForMGamma}
\end{equation}
then the fifth term in
brackets dominates all other terms. So, we are left with:
\begin{equation}
\mathds{P}(V_{\text{opt}}>c)\rightarrow 0
\end{equation}
super-exponentially fast as $N \rightarrow \infty$. Consequently:
\begin{center}
\emph{if the 
local dimension $d$ of a $N$-partite quantum system satisfies $d>mv(2m-1)^{2}$, then, for large 
$N$, the vast majority of pure states will not violate any Bell inequality with 
{bounded} coefficients to any significant degree.}
\end{center}

{
On the other hand, assume that $N\geq 3$ and that all parameters, except $d$ are 
fixed. Hence, as $d$ becomes arbitrarily 
large, we also see from Eq.~\eqref{cotadetalhada} that the probability \cris{of 
finding a violation} goes to zero. Putting into words:
\begin{center}
\emph{for any $N\geq 3$, if the local dimension $d$ is large enough, the vast majority of pure 
states will not violate any Bell inequality with bounded coefficients to any significant degree.}
\end{center}

To sum up, we can formally state:
\begin{corollary}
Let $\ket{\psi} \in \left( 
\mathds{C}^d \right)^N$ be a unit vector distributed according to the uniform 
measure in the sphere $S_{2d^{N}-1}$ of $\left( \mathds{C}^d \right)^N$. Given 
integers $N\geq 2$ and $d\geq2$, and given $b>0$ the following statements below hold 
true:
\begin{itemize}
 \item [a)]  If $d$ the local dimension satisfies $d>vm(2m-1)^{2} $, then 
\begin{equation}\mathds{P}(V_{\text{\normalfont{opt}}}>c)\rightarrow 0, 
\,\,\, \mbox{as}\,\,\, N \rightarrow \infty.\end{equation}
 \item [b)] If $N$ the number of parts satisfies $N \geq 3$, then:
\begin{equation}
\mathds{P}(V_{\text{\normalfont{opt}}}>c)\rightarrow 0, \,\, \mbox{as} \,\,\, 
d\rightarrow \infty.
\end{equation}
\end{itemize}
\label{coro:CorollaryForPositiveBellIneqs}
\end{corollary}
}

So, in spite the fact that 
typically any $N$-partite pure state, with large $N$ and/or $d$, is highly entangled, the 
typical value of the Bell violation is extraordinarily small. 


\section{The Proof}
\label{sec:TheProof}

\subsection{Idea of the proof}

We first decompose the event $\{V_{\text{\normalfont{opt}}}>c\}$ into a union of pieces 
corresponding to each possible choice of POVM's and Bell inequality coefficients, as well as 
in violations from above and violations from below. This gives:
\begin{align}
\mathds{P}\left(\ket{\psi}:  \sup_{\substack{A \in 
\mc{A} \\ T \in \mc{T}_{b} }}  |Q(\ket{\psi},T,A)|   
 > c\right)&=\mathds{P}\left(\left[\bigcup_{\substack{A \in 
\mc{A} \\ T \in \mc{T}_{b} }}  
\{\ket{\psi}:Q(\ket{\psi},T,A) > c\}\right]\bigcup \left[\bigcup_{\substack{A \in 
\mc{A} \\ T \in \mc{T}_{b} }}  
\{\ket{\psi}:Q(\ket{\psi},T,A) <-c\}\right] \right) \label{EqSimplea}\\
&=2\mathds{P}\left(\bigcup_{\substack{A \in 
\mc{A} \\ T \in \mc{T}_{b} }}  
\{\ket{\psi}:Q(\ket{\psi},T,A) > c\}\right).
\label{eq:EqSimple}
\end{align}
The first equality comes from the fact that the events considered are the same. In the disjoint in 
the \emph{r.h.s} of Eq.~\eqref{EqSimplea}, the fact that $Q$ is linear in $T$ guarantees that 
the two events have equal measure, which gives the second equality. Even though the event 
appearing on~\eqref{eq:EqSimple} is given by an 
infinite union of sets, the strategy is to replace it by a \emph{finite} union without, however, 
changing \cris{too} much its probability. This replacement is done trough suitable 
$\epsilon$-nets for 
$\mc{A}$ and $\mc{T}$. We then just 
apply the union bound, together with a uniform bound on the 
probability distribution for the degree of violation for any \emph{fixed} inequality 
with functional $T$ 
and measurement settings $A$. This strategy results in the bound of Theorem~\ref{Thm:TheTheorem} 
which has two terms as distinguished below:
\begin{align}
 \mathds{P}(V_{\text{\normalfont{opt}}}>c) \leq  
 4   
 \underbrace{\left[ \frac{8bNd^{2}(mv)^{N}}{\delta} + 2 
\right]^{mvNd^{2}}}_{(A)}  
\times \underbrace{\text{exp}{\left(-\frac{2d^{N}(c-\delta-1)^2}{36\pi^{2}(2m-1)^{2N-2}}
\right) } } _ { (B) }  ,
\label{eq:MainInequalityIdeaOfTheProof}
\end{align}
The term $(A)$ in Eq.~\eqref{eq:MainInequalityIdeaOfTheProof} is just an estimate of the number 
of points of the $\epsilon-$net. Term $(B)$ is the uniform bound mentioned above. It results 
basically from Lévy's 
lemma~\cite{MLedoux01} together with a pair of results on the smoothness of the $Q$ 
function.

\subsection{Some lemmas}
\label{subsec:GeneralArguments}

Note that $Q(\ket{\psi}, A, T )$ is a function of three variables: $i)$ $\ket{\psi}$, belonging to 
a sample space, $ii)$ $A \in \mc{A}$, which can be conveniently seen as an 
element of a metric space, as we are going to discuss later on subsequent subsections, and $iii)$ 
$T \in \mc{T}$ also belonging to a suitable metric set. 

The first step is show the possibility of replacing the infinite sets 
$\mc{A}$ and $\mc{T}$ by finite ``representative'' subsets of them, without changing much the 
optimal violation. This can be made precise by the following lemma.

\begin{lemma}\label{lemma1}
Let $\varphi:\Omega\times M\rightarrow \mathds{R}$ be such that $\Omega$ is a
sample space and $M$ a metric space. For a given $\delta>0$, suppose that there 
exists a finite $M_{\delta}^{\prime}\subseteq M$ such
that for every $x\in M$ there is $x'\in M_{\delta}^{\prime}$ where 
$|\varphi(\omega,x)-\varphi(\omega,x')|<\delta$ for every
$\omega\in\Omega$. Then, for every $c>0$, with $c-\delta>0$, we have 
\begin{equation}\mathds{P}\left(\omega:
\sup_{x\in M}\{\varphi(\omega,x)\}>c\right)\leq \sum_{x'\in
M_{\delta}^{\prime}}\mathds{P}\left(\omega:\varphi(\omega,x')>c-\delta\right).
\label{cota01}
\end{equation}
\end{lemma}
\begin{proof}

Take $\omega_*\in \Omega$ such that $\sup_{x\in M}\{\varphi(\omega_*,x)\}>c$. There exists then 
a $x\in M$ with $\varphi(\omega_*,x)>c$. Replacing $x$ by some $x^\prime\in M^\prime_\delta$ with 
$|\varphi(\omega_*,x) - \varphi(\omega_*,x^\prime)|<\delta$, we obtain 
$\varphi(\omega_*,x^\prime)>c-\delta$. Therefore, we have the following set inclusion.
\begin{equation}\{ \omega:\sup_{x\in M}\{\varphi(\omega,x)\}>c \} \subseteq \bigcup_{x'\in 
M_{\delta}^{\prime} }\{\omega:\varphi(\omega,x')>c-\delta\}. \label{ineqnec}\end{equation}The lemma 
now follows from a union bound on the set on the right hand side of the inclusion.\end{proof}

For our problem we will construct such a finite set based on an 
$\epsilon$-net~\cite{HE87,BrietVidick13} of a hypercube. Letting 
$||x||_{\infty}=\text{max}_{i=1,..n}|x_{i}|$,
for $x=(x_{1},...,x_{n})\in \mathds{R}^{n}$, we can state the following:
\begin{lemma}\label{lemma2}
Given $n\in\mathds{N}$, a subset $X$ of $[-1,1]^{n}$ and $0<\epsilon<1$, 
there
exists a finite subset $N_{\epsilon}\subset X$ such that for every 
$x\in X$ 
there exists $x'\in N_{\epsilon}$ with $||x-x'||_{\infty}<\epsilon.$ Moreover, 
\begin{equation}
|N_{\epsilon}|<\left(\frac{2}{\epsilon}+2\right)^{n}.
\end{equation}
\end{lemma}
\begin{proof}
Let $l=\lceil 1/\epsilon\rceil$ be the smallest integer greater than 
$1/\epsilon$. We can see the whole hypercube as the union of small hypercubes
\begin{equation}
h_{j_{1},...,j_{n}}=\prod_{k=1}^{n}\left[\frac{j_{k}}{l},\frac{j_{k}+1}{l}\right],
\end{equation}
where
$j_{k}\in\{-l,-l+1,...,-1,0,1,2,...,l-1\}$. Note that their edges have length $1/l$, so 
two points
$a,b$ inside them satisfy $||a-b||_{\infty}\leq 1/l=1/\lceil
1/\epsilon\rceil\leq \epsilon$. Whenever $X\bigcap h_{j_{1},...,j_{n}}\neq
\emptyset$ take a single point $x_{j_1,...,j_n}$ of $X$ in that 
intersection, and let $N_{\epsilon}$ be the collection
of such points. It is indeed finite and its number is at most the total number of small 
hypercubes,
$|N_{\epsilon}|\leq (2l)^{n}< (2/\epsilon+2)^{n}.$ Take now arbitrary 
$x\in X$. Since
$X\subseteq [-1,1]^{n}$,  $x$ belongs to some small hypercube 
$h_{j_{1}^{\prime},...,j_{n}^{\prime}}$, for some choice of indexes 
$\{j_{1}^{\prime},...,j_{n}^{\prime}\}$. If we take the single 
$x_{j_{1}^{\prime},...,j_{n}^{\prime}}\in
N_{\epsilon}$ in that hypercube, we have 
$||x-x_{j_{1}^{\prime},...,j_{n}^{\prime}}||\leq \epsilon$, since both $x$ 
and $x_{j_{1}^{\prime},...,j_{n}^{\prime}}$
belongs to $h_{j_{1}^{\prime},...,j_{n}^{\prime}}$.
\end{proof}

\begin{figure}[h]
 \includegraphics[scale=0.48]{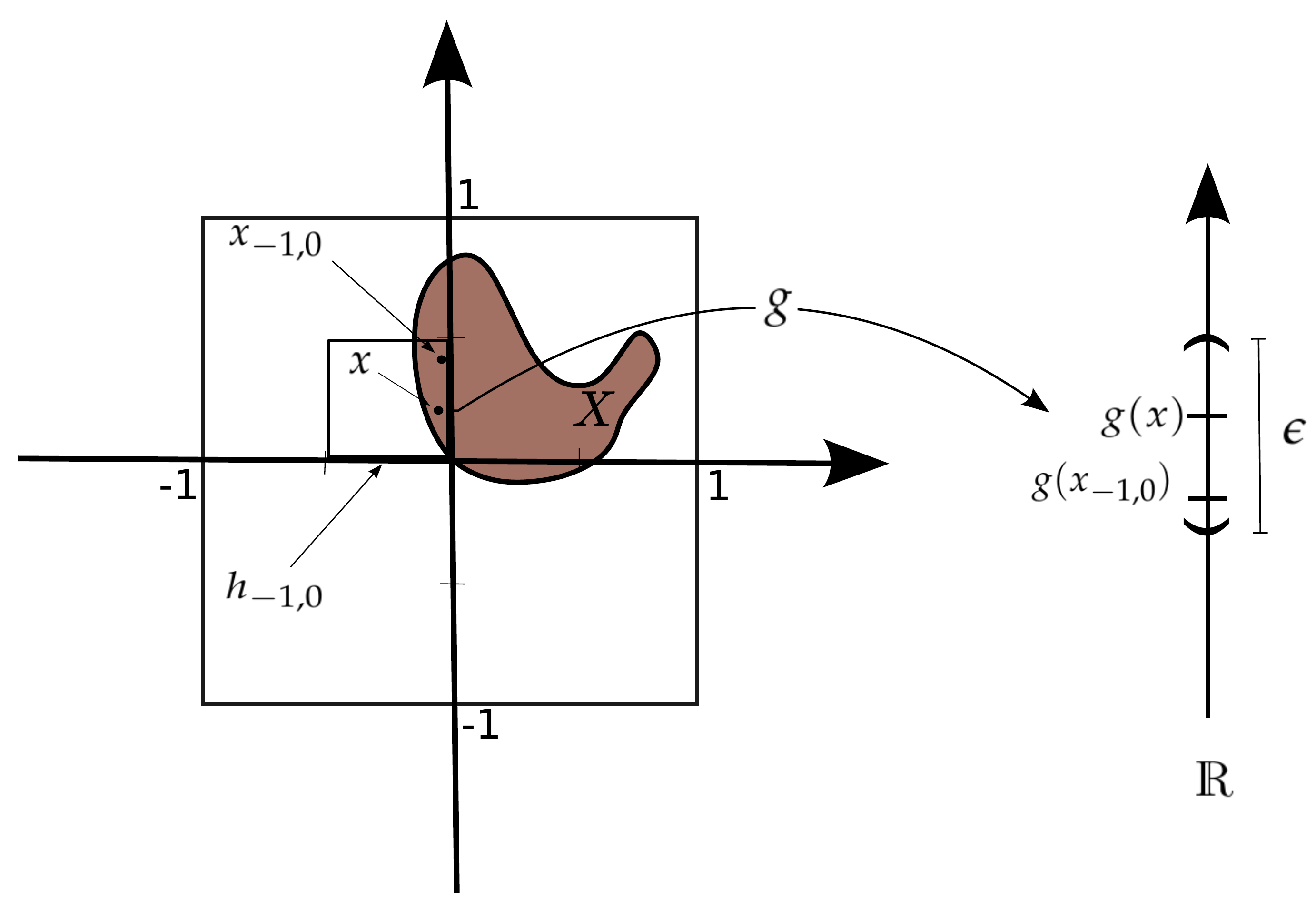}
 \caption{Schematic drawing of $X \subset [-1,1]^{2},$ and a small hypercube 
$h_{-1,0}$ intersecting it with $\epsilon=\frac{1}{2}$. Here we are showing 
an arbitrary choice of a point $x_{-1,0} \in N_{\epsilon}$, given $x \in X$, 
such that $|g(x)-g(x_{-1,0})| < \epsilon$ (colours online).} \label{figurabonita}
\end{figure} 

With Lemma~\ref{lemma2} in our hands, we have immediately (see 
Fig.~\ref{figurabonita}):
\begin{corollary}\label{lemma3}
Suppose that for $g:X\rightarrow \mathds{R}$, where $X\subseteq [-1,1]^{n}$, there 
exists  
$\lambda>0$ such that
$|g(x)-g(x')|\leq \lambda||x-x'||_{\infty}$ for all $x,x'\in X$. 
Given 
$\delta>0$ there exists a set $N_{\delta/\lambda}\subset X$ with
number of elements bounded by $(2\lambda/\delta+2)^{n}$ such that for 
every 
$x\in X$, there is $x'\in N_{\delta/\lambda}$ with $|g(x)-g(x')|< 
\delta$. 
\end{corollary}

Finally, we shall use L\'evy's Lemma (see~\cite{MLedoux01}):
\begin{lemma}
For every $\epsilon>0$, $D \geq 1$ integer and $F:S_{D}\rightarrow
\mathds{R}$ a real-valued function with Lipschitz constant $\Lambda$ (w.r.t the 
Euclidean distance),
the following inequality holds true:
\begin{equation}
\mathds{P}(F-\mathds{E}[F]>\epsilon)\leq 2 
e^{-\frac{(D+1)\epsilon^2}{9\pi^3\Lambda^2}},
\end{equation}
where $\mathds{P}$ denotes the uniform probability measure on the $D-$dimensional 
unity 
sphere 
$S_{D}$.
\end{lemma}

\section{Putting all together}\label{sec:PuttingAllTog}

Now, take a function $\varphi$ as in Lemma~\ref{lemma1}. Assume that 
$M=[-1,1]^{n}$ and there exists a
$\lambda>0$ such that $|\varphi(\omega,x)-\varphi(\omega,x')|\leq 
\lambda||x-x'||_{\infty}$ 
for all $\omega\in\Omega$ and $x,x'\in [-1,1]^{n}$. Given $c>\delta>0$, we 
can use Lemma~\ref{lemma1} and Corollary~\ref{lemma3} to 
show that:
\begin{equation}
 \mathds{P}\left(\omega:\sup_{x\in[-1,1]^{n}}\{\varphi(\omega,x)\}>c\right)\leq 
(2\lambda/\delta+2)^{n}\max_{x'\in 
N_{\lambda/\delta}}\mathds{P}\left(\omega:\varphi(\omega,
x')>c-\delta\right). 
\label{bound0}
\end{equation}

Let us assume additionally that $\Omega=S_{D}$ for some $D\in\mathds{N}$, that 
$\mathds{E}[\varphi(\cdot,x)]\leq 1$ for all $x\in M$, and that 
there exists non-negative $\Lambda \in \mathds{R}$ such that 
\begin{equation}
|\varphi(\omega,x)-\varphi(\omega',x)|\leq \Lambda|\omega-\omega'|, \,\, \mbox{for all} 
\,\, x\in 
[-1,1]^{n},\omega,\omega'\in S_{D}, 
\end{equation}
with 
$|\omega-\omega'|$ being the Euclidean distance on the sphere. That is, $\varphi$ 
is
Lipschitz with with respect to $\omega$, with Lipschitz constant $\Lambda$ independent 
of $x$. These assumptions
allow us 
to apply L\'evy's Lemma, and bound the probability appearing on the 
\emph{r.h.s} 
of Equation~\eqref{bound0} in order to get, for all $c>\delta+1$, the bound:
\begin{equation}
 \mathds{P}\left(\omega:\sup_{x\in[-1,1]^{n}}\{\varphi(\omega,x)\}>c\right)\leq 
2(2\lambda/\delta+2)^{n} 
e^{-\frac{(D+1)(c-\delta-1)^2}{9\pi^3\Lambda^2}}.
\label{bound1}
\end{equation}
We now come back to the setting of Theorem~\ref{Thm:TheTheorem}. We apply 
\eqref{bound1} to the events in Eq.~\eqref{eq:EqSimple}. We set $\varphi=Q$, 
$\omega=\ket{\psi}$ and $x=(T,A)$, that is, $M=\mc{T}\times\mc{A}$. The set  
of pure states $\ket{\psi}$ is a sphere $S_{D}$ with $D=2d^{N}-1$. Moreover, an 
upper bound for 
the mean value
of 
$Q$ follows straightforwardly from the fact that the  
expectation $\mathds{E}[Q(\ket{\psi},T,A)]$
corresponds to replacing $\ket{\psi}$ with the maximally mixed state, with which the Bell 
inequality is certainly satisfied:
\begin{equation}
 \mathds{E}[Q(\ket{\psi},T,A)] \leq 1.
\end{equation}
In order to apply our tools, we will parametrize 
$\mc{T}\times\mc{A}$ by a subset of $[-1,1]^{n}$, for some $n$, and find $\lambda$ and 
$\Lambda$ as 
described 
above. We will then obtain the bound:
\begin{equation}\label{eqchave}
\mathds{P}\left[\ket{\psi}:\sup_{\substack{A \in 
\mc{A} \\ T \in \mc{T} }}  Q(\ket{\psi},T,A)  
>c\right]\leq 4(2\lambda/\delta+2)^{n} 
e^{-\frac{(2d^{N})(c-\delta-1)^2}{9\pi^3\Lambda^2}}.
\end{equation}

In the following subsections we $i)$ show the existence, and $ii)$ obtain an 
estimative for the values of $n,\lambda,\Lambda$, which, inserted on the 
above expression, provides the proof of Theorem~\ref{Thm:TheTheorem}.

\subsection{Bounding the Lipschitz constant of $Q$ (estimating $\Lambda$)}
\label{subsec:BoundingLipschitz}
In order to show that the functional $Q$ is Lipschitz with respect to $\ket{\psi}$, 
firstly note that:
%
%
\begin{subequations}
\begin{align}
\forall \ket{\psi},\ket{\psi^{\prime}}: \vert 
Q(\ket{\psi},T,A) - Q(\ket{\psi'},T,A) \vert 
&= \vert \bra{\psi} 
\mathcal{B}_{T,A}\ket{\psi} - \bra{\psi^{\prime}} 
\mathcal{B}_{T,A}\ket{\psi^{\prime}} \vert  \label{eqlip1}  \\
&= \left| \text{Tr} \left[ \left( \sum_{\substack{j_1,...,j_N \\ 
o_1,...o_N}}T_{\substack{j_1,...,j_N \\ 
o_1,...o_N}} \bigotimes_{i=1}^{N} \Pi_{j_i,o_i}^{i}\right) 
(\ket{\psi}\bra{\psi}-\ket{\psi^{\prime}}\bra{\psi^{\prime}}) \right]  
\right| \label{eqlip3}\\
&\leq \left\|  \sum_{\substack{j_1,...,j_N \\ 
o_1,...o_N}}T_{\substack{j_1,...,j_N \\ 
o_1,...o_N}} \bigotimes_{i=1}^{N} \Pi_{j_i,o_i}^{i} \right\|    
         \left\|    
\ket{\psi}\bra{\psi}-\ket{\psi^{\prime}}\bra{\psi^{\prime}}   \right\|_{1} \label{eqlipc} \\
& =\left\| \mathcal{B}_{T,A} \right\|\left\|    
\ket{\psi}\bra{\psi}-\ket{\psi^{\prime}}\bra{\psi^{\prime}}   \right\|_{1} 
\label{eqlip4} \\
&\leq 2(2m-1)^{N-1} \left\| 
\ket{\psi}-\ket{\psi^{\prime}}   
\right\|_{2}. \label{eqlip5}
\end{align}
\end{subequations}
Here $\Vert \mathcal{B}_{T,A} \Vert$ stands for the usual operator norm, and 
 \eqref{eqlipc} follows from von Neumann's trace inequality~\cite{Mirsky75} 
plus H\"older's inequality. The bound for the Bell operator 
in~\eqref{eqlip5} follows from the fact that it is self-adjoint and Eq.~\eqref{cotaloubenets}. 
Therefore, a 
Lipschitz 
constant for $Q$ is 
\begin{equation}
\Lambda = 2(2m-1)^{N-1}.
\end{equation}

\subsection{Variation of $Q$ with $\mc{A}$ and $\mc{T}$ (finding $n$ 
and $\lambda$)}
\label{subsec:FindingNandLambda}

\subsubsection{Variation of $Q$ with measurements $A$}

Now we proceed to estimate how the function $Q(\psi,T,A)$ varies when we change the local operators 
describing the measurements at each site, for fixed linear functional $T$. We will describe this set 
of operators as a subset of an appropriate hypercube. 

Given two families 
$A,\tilde{A} \in \mc{A}$ of local operators, one has:
\begin{subequations}
\begin{align}
|Q(\psi,T,A) - Q(\psi,T,\tilde{A}) | & = | \bra{\psi} 
\mathcal{B}_{T,A}\ket{\psi} - \bra{\psi} 
\mathcal{B}_{T,\tilde{A}}\ket{\psi} | \label{deltaA2} \\
&= \left| \sum_{\substack{j_1,...,j_N \\ 
o_1,...,o_N}} T_{\substack{j_1,...,j_N \\ 
o_1,...,o_N}} \bra{\psi} \bigotimes_{k=1}^{N}(\Pi_{j_k,o_k}^{k} - 
\tilde{\Pi}_{j_k,o_k}^{k}) \ket{\psi} \right| 
\label{deltaA3} \\
& \leq \sum_{\substack{j_1,...,j_N \\ 
o_1,...,o_N}} \left| T_{\substack{j_1,...,j_N \\ 
o_1,...,o_N}} \right| N \sup_{k} \Vert \Pi_{j_k,o_k}^{k} - 
\tilde{\Pi}_{j_k,o_k}^{k} \Vert    \label{deltaA4} \\
& \leq bN(mv)^{N} \sup_{k,j_k,o_k} \Vert \Pi_{j_k,o_k}^{k} - 
\tilde{\Pi}_{j_k,o_k}^{k} \Vert. \label{deltaA5}
\end{align}
\end{subequations}
Note that in \eqref{deltaA4} we use that all POVM elements have norm bounded by one. In
\eqref{deltaA5} we use the fact that $\left| T_{j_1,...,j_N |
o_1,...,o_N} \right| \leq b$, for all $j_1,...j_N$ and 
$o_1,...,o_N$. 

Now, defining $D(A,\tilde{A})= \sup_{k,j_k,o_k} 
\Vert 
\Pi_{j_k,o_k}^{k} - 
\tilde{\Pi}_{j_k,o_k}^{k} \Vert$, we have:
\begin{equation}\label{eqlipparaQ1}
|Q(\psi,A,T) - Q(\psi,\tilde{A},T)| \leq 
bN(mv)^{N}D(A,\tilde{A}).
\end{equation}

Recall that the set of POVM operators in a given local Hilbert space has
(real) dimension $d^{2}$, the same for the set of Hermitian operators acting 
on that space. Indeed, if we fix an orthonormal basis
$\{\ket{i}\}_{i=1}^{d}$, we can write a POVM element
 as $\Pi=\sum_{i,j=1}^{d}\alpha_{ij}\ket{i}\bra{j}$ for some 
complex coefficients $\alpha_{ij}$. Collecting the real and imaginary parts 
of these coefficients, and using that $\Pi$ is self-adjoint, we get 
$d^{2}$ independent real parameters, which we can align on a 
$d^{2}-$dimensional vector $(r_1,...,r_{d^2})$, where $|r_{i}|\leq 1$ for 
$i=1,...,d^{2}$, since $0 \leq \Pi\leq I$. Describing each POVM element 
$\Pi^{k}_{j_{k},o_{k}}$ in this way, we will have a total of $d^{2}mvN$ 
real parameters taking values in the interval $[-1,1]$. In this way, we have defined a
mapping that takes each $A$ to a vector 
$\vec{r}(A)\in [-1,1]^{d^{2}mvN}$ in the $d^{2}mvN$-dimensional hypercube. We can moreover define a norm 
$\Vert \Pi \Vert_{\text{max}}=\max_{i\in 
[d^{2}]}\{ | r_i | \}$ such that:
\begin{equation}\label{normas}
 \Vert \Pi \Vert \leq 2d^{2}\Vert \Pi \Vert_{\text{max}}, \,\, \forall 
\Pi,
\end{equation} 
and the distance $D_{\text{max}}^{\mc{A}}(A,\tilde{A})= 
\sup_{k,j_k,o_k} \Vert \Pi_{j_k,o_k}^{k} - \tilde{\Pi}_{j_k,o_k}^{k} 
\Vert_{\text{max}}$, so 
that $D_{\text{max}}^{\mc{A}}(A,\tilde{A})=||\vec{r}(A)-\vec{r} 
(\tilde{A})||_{\infty}$ hence 
\begin{equation}\label{normas}
D(A,\tilde{A})\leq 2d^2 D_{\text{max}}^{\mc{A}}(A,\tilde{A}).
\end{equation}  
Finally, from~\eqref{normas} and \eqref{eqlipparaQ1}, we obtain the following
inequality:
\begin{equation}\label{ineqrede}
 |Q(\psi,A,T) - Q(\psi,\tilde{A},T)| \leq 
2bN(mv)^{N}d^{2}D_{\text{max}}^{\mc{A}}(A,\tilde{A}).
\end{equation}

\subsubsection{Variation of $Q$ with coefficients $T$}
\label{subsec:FindingConstantForTf}

We now estimate how much the function $Q$ changes when one 
varies coefficients in $\mc{T}$, \emph{i.e.} fixed a pure state 
$\ket{\psi}$, and a choice of local measurements $A \in 
\mc{A}$ we would like to know the behaviour of $|Q(\psi,T,A) - 
Q(\psi,\tilde{T},A)|$ with $T$ and $\tilde{T}$. One has:

\begin{align}
|Q(\ket{\psi},T,A) - Q(\ket{\psi},\tilde{T},A) |  &= | \bra{\psi} 
\mathcal{B}_{T,A}\ket{\psi} - \bra{\psi} 
\mathcal{B}_{\tilde{T},A}\ket{\psi} | \label{deltaB2} \\
&= \left| \sum_{\substack{a_1,...,a_N \\ 
x_1,...,x_N}} \left(  T_{\substack{a_1,...,a_N \\ 
x_1,...,x_N}} - \tilde{T}_{\substack{a_1,...,a_N \\ 
x_1,...,x_N}} \right) \bra{\psi} \bigotimes_{k=1}^{N}\Pi_{a_k,x_k}^{k} \ket{\psi} 
\right| 
\label{deltaB3} \\
& \leq \sum_{\substack{a_1,...,a_N \\ 
x_1,...,x_N}} \left| T_{\substack{a_1,...,a_N \\ 
x_1,...,x_N}} - \tilde{T}_{\substack{a_1,...,a_N \\ 
x_1,...,x_N}}  \right|   \label{deltaB4} 
\\
& \leq (mv)^{N} \Vert T - 
\tilde{T} \Vert_{\max}\\
&=b(mv)^N D_{\text{max}}^{\mc{T}}(T,\tilde{T}),
\label{deltaB5} 
\end{align}
introducing $D_{\text{max}}^{\mc{T}}(T,\tilde{T}):=\frac{1}{b}||T-\tilde{T}||_{\text{max}}$ so that 
we can see the elements of $\mc{T}_{b}$ as a subset of the hypercube $[-1,1]^{(mv)^{N}}$.

\subsubsection{Finding $n$ and $\lambda$}

Finally, returning to~\eqref{eqchave}, we can set then $M=\mc{T}\times 
\mc{A}\subset [-1,1]^{d^2mvN}\times[-1,1]^{(mv)^N}=[-1,1]^{d^2mvN+(mv)^N}$, with the following 
notion of distance:
$$\text{Dist}[(T,A),(T',A')]:=\text{max}\{D_{\text{max}}^{\mc{A}}(A,A'),D_{\text{max}}^{\mc{T}}(T,
T') \} .
$$ 
We immediately see that we can define $n=d^2mvN+(mv)^N$. Moreover, we have:
\begin{align}
|Q(\ket{\psi},T,A) - Q(\ket{\psi},T',A')|&\leq |Q(\ket{\psi},T,A) - 
Q(\ket{\psi},T,A')|+|Q(\ket{\psi},T,A') - Q(\ket{\psi},T',A')|\\
&\leq 2bN(mv)^{N}d^{2}D_{\text{max}}^{\mc{A}}(A,A')+ b(mv)^{N}D_{\text{max}}^{\mc{T}}(T,T')\\
&\leq 4bN(mv)^{N}d^{2}\text{Dist}[(T,A),(T',A')].
\end{align}
The last inequality allows us to set $\lambda=4Nb(mv)^Nd^2$.


\section{Conclusion}
\label{sec:Conclusion}

We have shown that for any fixed correlation scenario $\Gamma=(N,m,v)$ there exists an upper 
bound (see Eq.~\ref{eq:MainInequality}) for the typical 
violation that an $N-$partite pure state with local dimension $d$ can exhibit. In particular, we 
have proved that, if the local dimension $d$ is large enough relative to the complexity of the Bell 
scenario~\cite{BrietVidick13,PGWPVJ08,PalazuelosYin15}, then significant violations become extremely 
rare as $N$ increases. More precisely, given a correlation scenario $\Gamma=(N,m,v)$, if 
\begin{equation}
\frac{d}{m(2m-1)^{2}} > v, 
\label{eq:DependenceLocalDimension}
\end{equation}
then the probability of finding any significant violation of a Bell inequality whose 
coefficients are uniformly bounded~\footnote{An apparently very drastic behaviour like 
$M_{\Gamma}=\mc{O}((mv)^{N})$ is already enough for our purposes.} is extremely small for any 
$\delta>0$. That is: 
\begin{equation}\mathds{P}(V_{\text{\normalfont{opt}}} \geq 1+\delta)\rightarrow 0, \,\,\, 
\mbox{as} \,\, N \rightarrow \infty, \end{equation}
super-exponentially fast. This generalizes previous results by two of the present 
authors \cite{DO12}. In addition, also surprisingly, we also have shown that if the 
number of parts is greater than two, then as the local dimension $d$ goes to infinity, 
we also have the same behaviour for the probability of finding any significant 
violation, \emph{i.e.}: 
\begin{equation}
\mathds{P}(V_{\text{\normalfont{opt}}} \geq 1+\delta)\rightarrow 0, \,\,\, \text{as } d \rightarrow 
\infty. 
\end{equation}
We note that the typicality 
arguments used here and in \cite{DO12} are essential to several other results in 
quantum information 
theory~\cite{Palazuelos15,GHJPV2015,GLPV17,MEMetal15,BHH16,PSW06,BrietVidick13}) 

Remarkably, our result stands in contrast with the fact that  
entanglement becomes typically large in the same limits of $N\to \infty$ or $d\to \infty$ we 
are taking. This further suggest that entanglement, though necessary, is not sufficient to explain 
non-local effects in quantum physics.

We now argue that the dependance of our result on the local dimension $d$ is an essential feature of this kind of problem. In \cite{GHJPV2015} the authors showed that, for correlation scenarios $\Gamma=(2,m,2)$, the typical behaviour of local correlations can be quite different depending on the value of $d/m$. Depending on how this ratio behaves, one may obtain 
(asymptotically in $m$) that correlation matrices do or do not display quantum effects. On the 
other hand, in our setting we are optimizing over all possible 
$POVM$'s. Therefore, it should be expected that in any contextuality scenario high 
violations of contextuality inequalities might be typical. If true, this is another significant difference between these two distinct
situations~\cite{CSW10,CSW14,Amaral14,RDLTC14}.   

Our approach is very general, in that it 
encompasses a large class of Bell inequalities. However, if arbitrary coefficients are allowed, it remains open whether a similar result would hold. Based on recent work~\cite{GLPV17} we believe that for general 
correlation scenarios $\Gamma=(N,m,v)$ it is possible to get rid of the uniformly 
boundedness requirement.  

To conclude, we point out that our result also has implications in the 
context of the classical-to-quantum transition problem. As discussed firstly by 
Pitowsky~(see~\cite{Pitowsky91}), there is an apparent contradiction between the fact that, on the one hand, multipartite pure
 quantum states rarely admit LHV models, since they always violate some Bell 
inequality~\cite{Gisin91}; and, on the other hand,
the fact that in the macroscopic world such models are actually the rule. One way out of this is to invoke decoherence and claim that actual macroscopic systems should be 
described by highly mixed states, so that non-local correlations are not visible. Our result 
allows us to explain the same phenomenon from the (static) perspective of experimental feasibility. Indeed, even if a pure 
state is a good description for the macroscopic system, a typical state requires very intricate and sharp Bell test experiments in order to detect non-local 
correlations. This seems to be inconceivable in practice~\cite{Aspect75,Delft-free-15,Delft-free-16}.


\begin{acknowledgments}
We would like to thank Marco T\'ulio de Quintino, Gabriela Lemos, Rafael Chaves, 
Barbara Amaral and Marcelo Terra 
Cunha for useful comments and stimulating
discussions. CD thanks the International Institute 
of Physics for its support and hospitality. CD and RD thank CNPq, CAPES and FAPEMIG 
for financial support during the execution of this project. RO aknowledges the 
financial support of
Bolsa de Produtividade  em Pesquisa
from CNPq.
\end{acknowledgments}


\appendix

\section{An upper bound for the typical violation}
\label{sec:cotaprob}

Theorem~\ref{Thm:TheTheorem} gives us an upper bound for the typical 
violation for any Bell Scenario under consideration. In this appendix we will 
show how to rearrange and bound the terms present in the 
Eq.~\eqref{eq:MainInequality} in order to obtain a more reasonable inequality 
for our purposes.

Indeed:
\begin{subequations}
 \begin{align}
 \mathds{P}(V_{\text{\emph{opt}}}>c)& \leq 4  
 \left[\frac{8bN(mv)^Nd^2}{\delta}+2\right]^{mvNd^{2}+(mv)^N} \times 
e^{-\left(\frac{2d^{N}(c-\delta-1)^2}{36\pi^{2}(2m-1)^{2N-2}}\right)}, \label{eqdisc1} 
\\
& \leq 4e^{\left( mvNd^2 + (mv)^{N} \right)\log\left( 
\frac{16bN(mv)^{N}d^{2}}{\delta}
\right)}\times 
e^{-\left(\frac{2d^{N}(c-\delta-1)^2}{36\pi^{2}(2m-1)^{2N-2}}\right)} \label{eqdisc2} 
\\
& = 4 \exp \left\lbrace mvN^{2}d^{2}\log(mv) 
+mvNd^{2}\log\left(\frac{16bNd^{2}}{\delta} 
\right)  \right. + (mv)^{N}\log\left(\frac{16Nd^{2}}{\delta}\right) + 
N(mv)^{N}\log(bmv)
\nonumber \\
& \left.     
-\frac{(c-\delta-1)^{2}(2m-1)^2}{18\pi^{2}}\left[\frac{d} { (2m-1)^2 
} 
\right]^{N}\right\rbrace. \label{cotaajeitada}
 \end{align}
\end{subequations}
In the second term of~\eqref{eqdisc2} we used the bound 
$\frac{8bNd^2(mv)^N}{\delta}+2 \leq \frac{16Nd^2(mv)^N}{\delta}$, since 
$\delta$ is small, and $Nd^2(mv)^N>1$. In subsequent equations, we  
just rewrote some terms in order to facilitate comparison.

\bibliography{biblio}
\end{document}